\documentclass[envcountsect, envcountsame]{llncs}

\usepackage{makeidx}
\usepackage{textcomp}
\usepackage{amsmath,amssymb}
\usepackage{graphicx,graphics}
\usepackage{psfrag}
\usepackage{epsfig}

\newcommand{\N}{\ensuremath{\mathbb{N}}}

\newcommand{\luar}{\ensuremath{\leftarrow_{\rm u.a.r.}}}
\newcommand{\poly}{\ensuremath{{\textup{\rm poly}}}}

\newcommand{\vol}[3]{\ensuremath{{\rm vol}^{(#1)}(#2,#3)}}
\newcommand{\ball}[3]{\ensuremath{B^{(#1)}_{#2}(#3)}}

\newcommand{\schoening}{\texttt{\textup{Sch\"oning}}}
\newcommand{\dantsin}{\texttt{\textup{det-search}}}

\newcommand{\csps}{\texttt{\textup{Sch\"oning}}}
\newcommand{\cspsrun}{\texttt{\textup{One-Sch\"oning-Run}}}
\newcommand{\cspsb}{\texttt{\textup{searchball}}}
\newcommand{\gcspsb}{\texttt{\textup{G-searchball}}}
\newcommand{\ignore}[1]{}

\usepackage{cancel}
\usepackage{algorithm}
\usepackage{algorithmic}

\newcommand{\BCSP}{\ensuremath{\textsc{Ball-CSP}}}

\begin{document}
\frontmatter
\pagestyle{headings}

\title{Using a Skewed Hamming Distance to Speed Up Deterministic
  Local Search} \author{Dominik Scheder}
\institute{Theoretical Computer Science, ETH Z\"urich\\
  CH-8092 Z\"urich, Switzerland\\
  \email{dscheder@inf.ethz.ch}}

\maketitle

\begin{abstract}
  Sch\"oning~\cite{Schoening99} presents a simple randomized algorithm
  for $(d,k)$-CSP problems with running time
  $\left(\frac{d(k-1)}{k}\right)^n \poly(n)$. Here, $d$ is the number
  of colors, $k$ is the size of the constraints, and $n$ is the number
  of variables. A derandomized version of this, given by Dantsin et
  al.~\cite{dantsin}, achieves a running time of
  $\left(\frac{dk}{k+1}\right)^n \poly(n)$, inferior to Sch\"oning's.
  We come up with a simple modification of the deterministic
  algorithm, achieving a running time of $\left(\frac{d(k-1)}{k} \cdot
    \frac{k^d}{k^d-1} \right)^n \poly(n)$.  Though not completely
  eleminating the gap, this comes very close to the randomized bound
  for all but very small values of $d$.  Our main idea is to define a
  graph structure on the set of $d$ colors to speed up local search.
  \end{abstract}

\section{Introduction}

Constraint Satisfaction Problems, short CSPs, are a generali\-zation of
both bool\-ean satisfiability and the graph $k$-colorability problem.  A
set of $n$ variables $x_1,\dots,x_n$ is given, each of which can take
a value from $[d] := \{1,\dots,d\}$. The values $1,\dots,d$ are
sometimes called the {\em colors}. Each coloring of the $n$ variables,
also called assignment, can be represented as an element of $[d]^n$. A
{\em literal} is an expression of the form $(x_i \ne c)$ for some $c
\in [d]$. A CSP {\em formula} consists of a conjunction (AND) of {\em
  constraints}, where a constraint is a disjunction (OR) of literals.
We speak of $(d,k)$-CSP formula if each constraint consists of at most
$k$ literals. Finally, $(d,k)$-CSP is the problem of deciding whether
a given $(d,k)$-CSP formula has a satisfying assignment.\\

In 1999, Uwe Sch\"oning~\cite{Schoening99} came up with an extremely
simple and elegant algorithm for $(d,k)$-CSP: Start with a random
assignment. If this does not satisfy the formula, pick an arbitrary
unsatisfied constraint. From this constraint, pick a literal uniformly
at random, and assign to its underlying variable a new value, again
randomly.  Repeat this reassignment step $O(n)$ times, where $n$ is
the number of variables in the formula. If the formula $F$ is
satisfiable, we find a satisfying assignment with probability at least
\ignore{$\left(\frac{k}{d(k-1)}\right)^n \frac{1}{\poly(n)}$}
$(k/(d(k-1)))^n / \poly(n)$. By repeating this procedure $(d(k-1)/k)^n
\poly(n)$ times, we obtain an exponential Monte Carlo algorithm for
$(d,k)$-CSP which we will call $\schoening$. Not long afterwards, in
2002, Dantsin, Goerdt, Hirsch, Kannan, Kleinberg, Papadimitriou,
Raghavan and Sch\"oning~\cite{dantsin} designed a {\em deterministic}
algorithm based on deterministic local search and covering codes.
This algorithm, henceforth called $\dantsin$, can be seen as an
attempt to derandomize Sch\"oning's random walk algorithm (actually
the authors cover only the case $d=2$, but everything nicely
generalizes to higher $d$). I say {\em attempt} because its running
time of $(dk/(k+1))^n / \poly(n)$
\ignore{$\left(\frac{dk}{k+1}\right)^n \poly(n)$} is worse than that
of $\schoening$. \\

Consider the following variant of $\dantsin$: Suppose $F$ is a
$(d,k)$-CSP formula on $n$ variables, with $d = 2^{\ell}$ being a
power of $2$. Replacing every $d$-ary variable by $\ell$ boolean
variables, we transform $F$ into a $(2,\ell k)$-CSP formula $F'$ over
$\ell n$ variables. We solve $F'$ using the original algorithm
$\dantsin$ for the boolean case. A quick calculation shows that this
already improves over the running time of $(dk/(k+1))^n / \poly(n)$.
This observation motivates a more systematic exploration of possible
ways to speed up $\dantsin$. The main contribution of this paper is a
modified $\dantsin$ algorithm, which achieves a significantly better
running time.  Both $\schoening$ and $\dantsin$ work by locally
exploring the {\em Hamming graph} on $[d]^n$, in which two assignments
are connected by an edge if they differ on exactly one variable. We
define a graph $G$ on the set $\{1,\dots,d\}$ of colors, thus
obtaining a different, sparser graph on $[d]^n$, the $n$-fold
Cartesian product $G^{\Box n}$: Two assignments are connected by an
edge if they differ on exactly one variable, and on that variable, the
two respective colors are connected by an edge in $G$. With $G=K_d$,
this is the Hamming graph on $[d]^n$. Taking $G$ to be the directed
cycle on $d$ vertices, it turns out that our modified deterministic
algorithm has a running time of
$$
\left(\frac{d(k-1)}{k} \cdot \frac{k^d}{k^d-1} \right)^n \poly(n) \
.$$ For $d \geq 3$, this is significantly better than $\dantsin$ and
comes very close to $\schoening$ except if $d$ is very small (in
particular, we do not improve the case $d=2$).  We hope that future
research will eventually lead to a complete derandomization. We
compare running times for some values of $d$ and $k$ (ignoring
polynomial
factors in $n$):\\

\begin{center}
\begin{tabular}{c||c|c|c}
  $(d,k)$ & $\schoening$ & $\dantsin$ & this paper\\ \hline 
  $(2,3)$ & $1.334^n$ & $1.5^n$ & $1.5^n$ \\ \hline
  $(3,3)$ & $2^n$   & $2.25^n$ & $2.077^n$ \\ \hline
  $(5,4)$ & $3.75^n$  & $4^n$ & $3.754^n$ 
\end{tabular}
\end{center}

The case of $(2,k)$-CSP, commonly called $k$-SAT, has drawn most
attention, in particular $3$-SAT. For $3$-SAT, $\schoening$ and
$\dantsin$ achieve a running time of $O(1.334^n)$ and $O(1.5^n)$,
respectively. This is still very close to the current records: By
combining $\schoening$ with a randomized algorithm by Paturi,
Pudl\'{a}k, Saks, and Zane~\cite{ppsz}, Iwama and
Tamaki~\cite{iwama-tamaki} achieved a running time of $O(1.3238^n)$.
Later, Rolf~\cite{rolf05} improved the analysis of their algorithm to
obtain the currently best bound of $O(1.32216^n)$.  The algorithm
$\dantsin$ has been improved as well, first to $O(1.481^n)$ by the
same authors, then to $O(1.473^n)$ by Brueggemann and
Kern~\cite{BK04}, and finally to the currently best deterministic
bound of $O(1.465^n)$ by myself~\cite{Scheder08}. Though we do not
improve the case $d=2$ in this paper, we hope that better
understanding of general $(d,k)$-CSP will lead to better algorithms
for $k$-SAT, as well.\\
\ignore{ A common feature of all $k$-SAT algorithms is that the
  exponential term in the running time converges to $2^n$ as $k$
  grows. In fact, the exponential term of the running time of all
  known algorithms for $(d,k)$-CSP converges to $d^n$ for growing $k$.
  It would be a breakthrough to find an algorithm
  for $k$-SAT running in time $O(1.999^n)$, for all values of $k$.\\
}

Another fairly well-investigated case is $k=2$ with $d$ being large.
In 2002, Feder and Motwani~\cite{feder-motwani} adapted a randomized
$k$-SAT algorithm by Paturi, Pudl\'{a}k and Zane~\cite{ppz} to
$(d,2)$-CSP, obtaining a running time of $(c_d d)^n$, with $c_d$
converging to $e^{-1}$ as $d$ grows. We see that the base of the
exponential term is proportial to $d$. A certain growth of the base
with $d$ seems inevitable: Recently, Traxler~\cite{Traxler08} showed
that an algorithm solving $(d,2)$-CSP in time $a^n$, with $a$ being
independent of $d$, could be used to solve $k$-SAT in {\em
  subexponential} time, i.e., $2^{o(n)}$. This would contradict the
 {\em exponential time hypothesis}~\cite{IPZ01}.

\subsection*{Organization of this paper}

In Section~\ref{section-algorithms} we describe $\schoening$ and
$\dantsin$, and analyze the running time of the latter. Although most
of the material of Section~\ref{section-algorithms} is
from~\cite{Schoening99} and~\cite{dantsin}, we chose to present it
here in order to make the paper self-contained. In
Section~\ref{section-g-ball}, we define a graph structure on the set
of colors, which changes the notion of distance on the set $[d]^n$.
Taking this graph to be the directed cycle on $d$ vertices yields a
significant improvement. In In Section~\ref{section-optimality}, we
show that choosing this graph is optimal.

\section{The Algorithms $\schoening$ and $\dantsin$}
\label{section-algorithms}

Sch\"oning's algorithm works as follows. We start with a random
assignment, and for $O(n)$ steps randomly correct it locally. By this
we mean choosing an arbitrary non-satisfied constraint $C$, then
choosing a literal $(x \ne c) \luar C$ (where $\luar$ means choosing
something uniformly at random), and randomly re-coloring $x$ with some
$c' \luar [d]\setminus \{c\}$.
\begin{algorithm}
\caption{\cspsrun($F$: a $(d,k)$-CSP) formula}
\label{csps}
\begin{algorithmic}[1]
  \STATE $\alpha \luar [d]^n$
  \FOR{$i=1,\dots,c n$}
  \STATE // {\em $c$ is a constant depending on $d$ and $k$, but not on $n$}
  \IF{$\alpha$ satisfies $F$} \RETURN $\alpha$
  \ELSE
  \STATE $C \leftarrow $ any constraint of $F$ unsatisfied by $\alpha$  
  \STATE $(x\ne c) \luar C$ \hspace{1cm}// {\em a random literal from $C$}
  \STATE $c' \luar [d] \setminus \{c\}$ \hspace{0.8cm} // 
  {\em choose a new color for $x$}
  \STATE $\alpha \leftarrow \alpha[x := c']$ \hspace{1.2cm} // {\em 
    change the coloring $\alpha$}
  \ENDIF
  \ENDFOR
  \RETURN \texttt{unsatisfiable}
\end{algorithmic}
\end{algorithm}

\begin{theorem}[\cite{Schoening99}]
  If $F$ is a satisfiable $(d,k)$-CSP formula on $n$ variables, then
  \cspsrun\ returns a satisfying assignment with probability at least
  $$\left(\frac{d(k-1)}{k}\right)^{-n} \frac{1}{\poly(n)}\ .$$
\end{theorem}

By repeating $\cspsrun$ \ignore{$\left(\frac{d(k-1)}{k}\right)^n
  \poly(n)$} $(d(k-1)/k))^n \poly(n)$, times, we find a satisfying
assignment with high probability. This yields a randomized Monte-Carlo
algorithm of running time $(d(k-1)/k))^n \poly(n)$,
\ignore{$\left(\frac{d(k-1)}{k}\right)^n \poly(n)$,} which we call $\csps$. \\

Let us describe the algorithm $\dantsin$ from~\cite{dantsin}. We
define a parametrized version of CSP, called \BCSP. Given a
$(d,k)$-CSP formula $F$, an assignment $\alpha \in [d]^n$ and an $r \in \N_0$,
does there exist a satisfying assignment $\beta$ such that
$d_H(\alpha,\beta)\leq r$? Here,
$$
d_H(\alpha,\beta) = | \{1 \leq i \leq n \ | \ \alpha(x_i) \ne \beta(x_i)\}|
$$
is the {\em Hamming distance}, and  
$$
\ball{d}{r}{\alpha} := \{\beta \in [d]^n \ | \ d_H(\alpha, \beta)\leq r\}
$$
is the {\em Hamming ball} of radius $r$ around $\alpha$. In other
words, \BCSP\ asks whether $B_r(\alpha)$ contains a satisfying
assignment. The algorithm \cspsb\ solves it in time $ (k(d-1))^r
\poly(n)$.
\begin{algorithm}
\caption{\cspsb(CSP formula $F$, assignment $\alpha$, radius $r$)}
\label{cspsb}
\begin{algorithmic}[1]
  \IF{$\alpha$ satisfies $F$} \RETURN \texttt{true}
  \ELSIF{$r = 0$} \RETURN \texttt{false}
  \ELSE
  \STATE $C \leftarrow $ any constraint of $F$ unsatisfied by $\alpha$  
  \FOR{$(x \ne c) \in C$}
  \FOR{$c' \in [d]\setminus c$}
  \STATE $\alpha' \leftarrow \alpha[x := c']$
  \IF{$\cspsb(F, \alpha', r-1) = \texttt{true}$} \RETURN \texttt{true}
  \ENDIF
  \ENDFOR
  \ENDFOR
  \RETURN \texttt{false}
  \ENDIF
\end{algorithmic}
\end{algorithm}
To show correctness, suppose the ball $\ball{d}{r}{\alpha}$ contains a
satisfying assignment $\beta$, and let $C$ be a constraint not
satisfied by $\alpha$. At least one literal $(x \ne c) \in C$ is
satisfied by $\beta$, and in one iteration of the inner for-loop, the
algorithm will change the assignment $\alpha$ to $\alpha'$ such that
$\alpha'(x) = \beta(x)$, and therefore $d_H(\alpha', \beta) =
d_H(\alpha,\beta)-1 \leq r-1$, and at least one recursive call will be
successful. The running time of this algorithm is easily seen to be at
most $(k(d-1))^r \poly(n)$, as each call causes at most $k(d-1)$
recursive calls (see Figure~\ref{csp-branching-1} for an
illustration), and
takes a polynomial number of steps itself.\\
\begin{figure}
  \begin{center}
    \epsfig{file=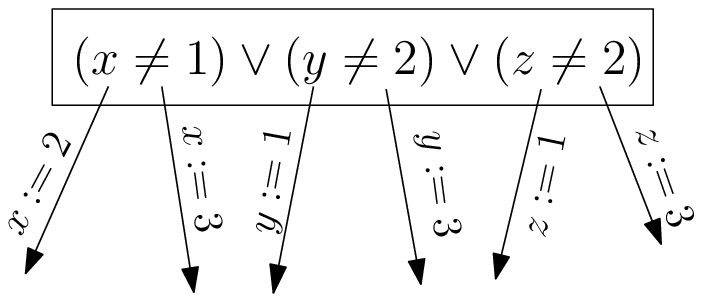,width=0.4\textwidth}
  \end{center}
  \caption{$\cspsb$ branching on a constraint of a $(3,3)$-CSP formula.}
  \label{csp-branching-1}
\end{figure}

\subsection*{Covering Codes}
How can we turn this algorithm into an algorithm for searching $[d]^n$
for a satisfying assignment? Suppose somebody gives us a 
set $\mathcal{C} \subseteq [d]^n$ such that
$$
 \bigcup_{\alpha \in \mathcal{C}} \ball{d}{r}{\alpha} = [d]^n \ ,
$$
i.e. a {\em code of covering radius $r$}. By calling
$\cspsb(F,\alpha,r)$ for each $\alpha \in \mathcal{C}$, we can decide
whether $[d]^n$ contains a satisfying assignment for $F$ in time
\begin{eqnarray}
|\mathcal{C}| (k(d-1))^r \poly(n) \ .
\label{running-time-1}
\end{eqnarray}
By symmetry of the cube $[d]^n$, the cardinality of
$\ball{d}{r}{\alpha}$ does not depend on $\alpha$, and we define
$\vol{d}{n}{r} := |\ball{d}{r}{\alpha}|$. The following lemma gives
bounds on the size of a covering code $\mathcal{C}$.
\begin{lemma}[\cite{dantsin}]
  For all $n,d,r$, every code $\mathcal{C}$ of covering radius $r$ has
  at least $\frac{d^n}{\vol{d}{n}{r}}$ elements. Furthermore, there is
  such a $\mathcal{C}$ with
  $$
  |\mathcal{C}| \leq \frac{[d]^n}{\vol{d}{n}{r}} \poly(n) \ ,
  $$
  and furthermore, $\mathcal{C}$ can be constructed deterministically
  in time $|\mathcal{C}| \poly(n)$.
\label{lemma-CC}
\end{lemma}
This lemma, together with (\ref{running-time-1}), yields a
deterministic algorithm solving $(d,k)$-CSP in time
$\frac{d^n}{\vol{d}{n}{r}} (k(d-1))^r \poly(n)$, and we are free to
choose $r$.  At this point, Dantsin et al. use the estimate
$\vol{2}{n}{r} = \sum_{i=0}^r {n \choose i} \geq 2^{nH(r/n)} /
\poly(n)$, where $H(x)$ is the binary entropy function (see
MacWilliams, Sloane~\cite{MacWilliams-Sloane77}, Chapter 10, Corollary
9, for example), but we prefer to derive the bounds we need ourselves,
first because the calculations involved are simpler, and second
because our method easily generalizes to the volume of more
complicated balls we will define in the next section.  We use
generating functions, which are a well-established tool for
determining the asymptotic growth of certain numbers (cf. the book
{\em generatingfunctionology}~\cite{gfology}).
\begin{lemma}
  For any $n,d \in \N$ and $x \geq 0$, there is an $r \in
  \{0,1,\dots,n\}$ such that
  $$\vol{d}{n}{r} \geq \frac{1}{n+1} 
  \frac{(1+(d-1)x)^n}{x^r} \ .$$
\label{lower-bound-binom}
\end{lemma}
\begin{proof}
  We write down the generating function for the sequence
  $\left({n \choose i}(d-1)^i\right)_{i=0}^{n}\ $:
  $(1 + (d-1)x)^n = \sum_{i=0}^n {n \choose i} (d-1)^i x^i$. This sum
  involves $n+1$ terms, the maximum being attained at $i = r$
  for some $i \in \{0,\dots,n\}$.
  Thus $(1 + (d-1)x)^n \leq (n+1) {n
    \choose r} (d-1)^r x^r$. Using 
  $\vol{d}{n}{r} \leq {n \choose r}(d-1)^r$ and re-arranging
  terms yields the claimed bound. \qed
\end{proof}

\begin{theorem}
  There is a deterministic algorithm solving $(d,k)$-CSP in 
  time $\left(\frac{dk}{k+1}\right)^n \poly(n)$.
\label{det-normal-running-time}
\end{theorem}
\begin{proof}
  Choose $x := (k(d-1))^{-1}$ and apply
  Lemma~\ref{lower-bound-binom}. The lemma gives us some $r \in
  \{0,\dots,n\}.$ With this radius, the running time 
  in (\ref{running-time-1}) is at most
  \begin{eqnarray*}
    & & \frac{d^n}{\vol{d}{n}{r}} (k(d-1))^r \poly(n)  = 
    \frac{d^n x^r (k(d-1))^r}{(1+(d-1)x)^n} \poly(n)\\
    & = & \left(\frac{d}{1+(d-1)\frac{1}{k(d-1)}}\right)^n \poly(n)
    = \left(\frac{dk}{k+1}\right)^n\poly(n) \ .
  \end{eqnarray*}
  \qed
\end{proof}
Let us summarize the algorithm $\dantsin$: It first constructs a code
of appropriate covering radius, then calls $\cspsb$ for every element
in the code. Its running time is larger than that of $\schoening$,
since \ignore{$\frac{dk}{k+1} \geq \frac{d(k-1)}{k}$} $dk/(k+1) \geq
d(k-1)/k$.  \ignore { For large $k$, the difference decreases, which
  is not surprising, as both running times converge to the trivial
  upper bound of $d^n$. For small $k$ and any $d$, the difference is
  significant. In the next section, we introduce a simple change to
  the deterministic algorithm that almost closes this gap, at least
  for all not too small values of $d$.  }

\ignore{
\section{A Cheap Trick}

Any $(4,k)$-CSP formula with $n$ variables can be encoded as a
$(2,2k)$-CSP formula with $2n$ variables: Just replace each $4$-ary
variable $x$ by two binary variables $x_1, x_2$, encoding the colors
$1,2,3,4$ by two bits in some way. A $k$-constraint involving $4$-ary
variables becomes a $2k$-constraint involving binary variables. In
other words, we obtain a $2k$-CNF formula. By
Theorem~\ref{det-normal-running-time}, we can solve $(2,2k)$-CSP
in time
\begin{eqnarray}
\left(\frac{4k}{2k+1}\right)^{2n} \poly(n) \ .
\label{eqn-4-using-cube}
\end{eqnarray}
For example, for $k=3$ this is $\left(\frac{144}{49}\right)^n \poly(n)
\approx 2.939^n \poly(n)$, better than the running time of $3^n
\poly(n)$ we obtain we directly applying
Theorem~\ref{det-normal-running-time} to a $(4,3)$-CSP formula. Not a great
improvement, but one that has been obtained in a surprisingly simple
way. Clearly, the same trick works whenever $d$ is a power of $2$:
\begin{theorem}
  Suppose $d = 2^\ell$ for some $\ell \in \N$. Then $(d,k)$-CSP can be
  solved in time $\left(\frac{2\ell k}{\ell k +1}\right)^{\ell n}
  \poly(n)$.
\label{encode-binary}
\end{theorem}
It is not difficult to check that $\left(\frac{2 \ell k}{\ell k +
    1}\right)^\ell < \frac{dk}{k+1}$ for any $\ell \geq 2$, and thus
Theorem~\ref{encode-binary} is an improvement over
Theorem~\ref{det-normal-running-time} whenever $d$ is a power of $2$
and at least $4$. This is surprising, since by replacing a
$2^\ell$-ary variable $x$ by $\ell$ binary variables
$x_1,\dots,x_\ell$, we are throwing away a lot of information: A
clause contains $x_i$ if and only if it also contains all
$x_1,\dots,x_\ell$, but our $(2,\ell k)$-CSP algorithm does not use
this fact at all. Two obvious question arise: First, can we improve
Sch\"oning's algorithm in a similar fashion? Second, what do we do if
$d$ is not a power of $2$? The first question is answered quickly:
Sch\"oning's running time on a $(2^\ell, k)$-CSP formula with $n$
variables is $\left(\frac{2^\ell (k-1)}{k}\right)^n \poly(n)$. Using
binary encoding, this becomes $\left(\frac{2 (\ell k - 1)}{\ell
    k}\right)^{\ell n} \poly(n)$.  Again, a simple calculation shows
that $\left(\frac{2 (\ell k - 1)}{\ell k}\right)^\ell >
\frac{d(k-1)}{k}$, and thus binary encoding actually makes
Sch\"oning's algorithm worse. The second question is more difficult:
If $d$ is not a power of $2$, say $d=3$, how can we encode our
variables? Using two binary variables to encode one ternary variable
is wasteful and does not lead to an improvement. It turns out that it
makes sense to again study $(4,k)$-CSP, but this time not applying any
explicit encoding in
binary.\\

Suppose our algorithm $\cspsb$ runs on a $(4,3)$-CSP formula and
encounters the unsatisfied constraint $(x \ne 1) \vee (y \ne 2) \vee
(z \ne 4)$.  It causes $k(d-1)=9$ recursive calls, exhausting all
possibilities to change the assignment at one variable occurring in
the constraint. See Figure~\ref{csp-branching-2} for an illustration.
\begin{figure}
  \begin{center}
    \epsfig{file=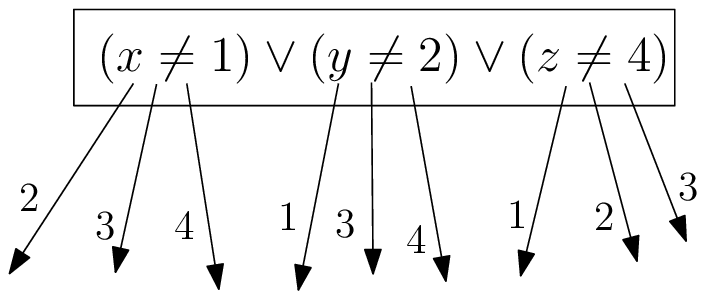,width=0.5\textwidth}
  \end{center}
  \caption{A typical branching of $\cspsb$ on a $(4,3)$-CSP formula.}
  \label{csp-branching-2}
\end{figure}
How does the branching look after applying binary encoding? Suppose we
encode the values $1$,$2$,$3$, and $4$ by $01$, $10$, $11$, and $00$,
respectively.  The first literal $(x \ne 1)$ translates into $(x_1 \ne
0) \vee (x_2 \ne 1)$, and the algorithm in one branch changes
$x_1$ to $1$, and in the other one changes $x_2$ to $0$. In other
words, it changes $x$ to $3$ and $x$ to $4$, but does not change $x$
to $2$ in any branch. The branching now looks like in
Figure~\ref{csp-branching-3}.
\begin{figure}
  \begin{center}
    \epsfig{file=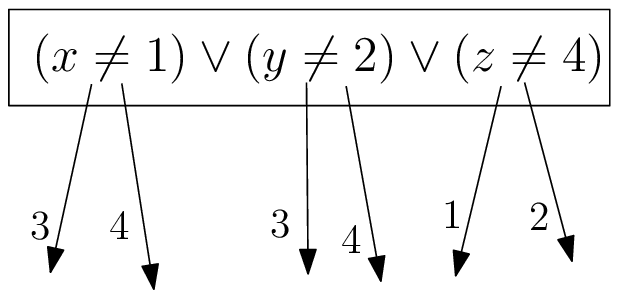,width=0.5\textwidth}
  \end{center}
  \caption{Branching on a $(4,3)$-CSP formula using binary encoding.}
  \label{csp-branching-3}
\end{figure}
Obviously we eliminated $3$ recursive calls. Rather than vieweing this
from the point of binary encoding, we can view it like this: the
colors $3$ and $4$ are neighbors of color $1$, but $2$ is not a
neighbor of $1$. Therefore, the first literal $(x \ne 1)$ causes no
call with $x$ being set to $2$. We only change colors to neighboring
colors.  The neighborhood relation between the colors $1$,$2$,$3$, and
$4$ can be represented by a graph, see Figure~\ref{C-4}:
\begin{figure}
  \begin{center}
    \epsfig{file=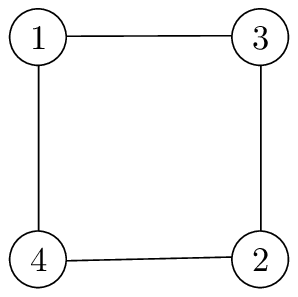,width=0.2\textwidth}
  \end{center}
  \caption{The neighborhood relation between the colors $1,2,3,4$ when
    using binary encoding.}
  \label{C-4}
\end{figure}

Reducing the number of neighbors clearly speeds up $\cspsb$. However,
as there is no free lunch, it reduces the number of elements at
distance $r$, i.e., the volume of the Hamming balls. Apparently, the
favorable effect is stronger, and we can improve upon the running
time. Since we are altering the notion of neighbors and thus of 
distance, we should define everything formally.
}

\section{$G$-Distance, $G$-Balls, and $\gcspsb$}
\label{section-g-ball}

Let $[d]$ be the set of colors, and let $G = ([d],E)$ be a (possibly
directed) graph. For two colors $c, c'$, we denote by $d_G(c,c')$ the
length of a shortest path from $c$ to $c'$ in $G$. If $G$ is directed,
this is not necessarily a metric, and therefore we rather call it a
{\em distance function}. It gives rise to a distance function on $[d]^n$:
For two assignments $\alpha, \beta \in [d]^n$, we define
\begin{eqnarray}
  d_G(\alpha, \beta) = \sum_{i=1}^{n} d_G(\alpha_i,\beta_i)  \ .
\end{eqnarray}
This is the shortest-path distance on the $n$-fold Cartesian product
$G^{\Box n}$. This distance induces the notion of balls
$\ball{G}{r}{\alpha} := \{\beta \in [d]^n \ | \ d_G(\alpha,\beta) \leq
r\}$, and of {\em dual balls} $\{\beta \in [d]^n \ | \
d_G(\beta,\alpha) \leq r\}$.  If $G$ is undirected, balls and dual
balls coincide, and for $G$ being $K_d$, the complete undirected
graph, $d_G$ is simply the Hamming distance.  If $G$ is
vertex-transitive (and possibly directed), the cardinality
$|\ball{G}{r}{\alpha}|$ does not depend on $\alpha$, and we define
$\vol{G}{n}{r} := |\ball{G}{r}{\alpha}|$. By double-counting, this is
also the cardinality of dual balls. In particular, a vertex-transitive
graph is {\em regular}. Let $\delta$ denote the number of edges
leaving each vertex in $G$. As before, we define a parametrized
problem: Given $F$, $\alpha$ and $r$, does $\ball{G}{r}{\alpha}$
contain a satisfying assignment?  Algorithm~\ref{gcspsb}, almost
identical to Algorithm~\ref{cspsb}, solves this problem in time
$(\delta k)^r \poly(n)$.
\begin{algorithm}
\caption{\gcspsb(CSP formula $F$, assignment $\alpha$, radius $r$)}
\label{gcspsb}
\begin{algorithmic}[1]
  \IF{$\alpha$ satisfies $F$} \RETURN \texttt{true}
  \ELSIF{$r = 0$} \RETURN \texttt{false}
  \ELSE
  \STATE $C \leftarrow $ any constraint of $F$ unsatisfied by $\alpha$  
  \FOR{$(x \ne c) \in C$}
  \FOR{all $c'$ such that $(c,c') \in E(G)$}
  \STATE $\alpha' \leftarrow \alpha[x := c']$
  \IF{$\cspsb(F, \alpha', r-1) = \texttt{true}$} \RETURN \texttt{true}
  \ENDIF
  \ENDFOR
  \ENDFOR
  \RETURN \texttt{false}
  \ENDIF
\end{algorithmic}
\end{algorithm}

\subsection{Covering Codes, Again}

Using the $G$-distance function instead of the Hamming distance also
induces the notion of covering codes. As before, $\mathcal{C}\subseteq
[d]^n$ is a code of covering $G$-radius $r$ if 
$$
\bigcup_{\alpha \in \mathcal{C}} \ball{G}{r}{\alpha} = [d]^n \ .
$$
The following lemma generalizes Lemma~\ref{lemma-CC} to arbitrary
vertex-transitive graphs $G$ on $d$ vertices. The proof does not
introduce any new ideas, and can be found in the appendix.
\begin{lemma}
  Let $d \geq 2$, and let $G$ be a vertex transitive graph on $d$
  vertices. For all $n$ and $0 \leq r \leq n$, every code
  $\mathcal{C}$ of covering $G$-radius $r$ has at least
  $d^n/ \vol{G}{n}{r}$ elements.  Furthermore, there is such a
  $\mathcal{C}$ with
  $$
  |\mathcal{C}| \leq \frac{[d]^n}{\vol{G}{n}{r}} \poly(n) \ ,
  $$
  and $\mathcal{C}$ can be constructed deterministically
  in time $|\mathcal{C}| \poly(n)$.
\label{lemma-G-CC}
\end{lemma}
By calling $\gcspsb(F,\alpha,r)$ for each $\alpha \in \mathcal{C}$,
we can solve $(d,k)$-CSP deterministically in  time
\begin{eqnarray}
\frac{d^n}{\vol{G}{n}{r}} (k\delta)^r \poly(n) \ ,
\label{G-runtime}
\end{eqnarray}
where we are free to choose any vertex-transitive graph $G$ and any
radius $r$. Let us reflect over (\ref{G-runtime}) for a minute.
Taking a graph with many edges results in balls of greater volume,
meaning a smaller $\mathcal{C}$ but spending more time searching each
ball. Taking $G$ to be rather sparse has the opposite effect.  What is
the optimal graph $G$ and the optimal radius $r$?

\subsection{Directed Cycles}

Let us analyze the algorithm using $G = C_d$, the directed cycle on
$d$ vertices. Clearly, $\delta = 1$, and therefore $\gcspsb$ runs in
time $k^r$. This is as fast as we can expect for any strongly
connected graph. What is $\vol{C_d}{n}{r}$?
\begin{lemma}
  For any $n,d \in \N$, and $x \geq 0$, there is an $r \in
  \{0,\dots,(d-1)n\}$ such that $$\vol{C_d}{n}{r} \geq
  \frac{(1+x+\dots+x^{d-1})^n}{x^r} \cdot \frac{1}{\poly(n)}\ .$$
\label{lower-bound-trinom}
\end{lemma}
\begin{proof}
  Define $T(n,s):= \vol{C_d}{n}{s} - \vol{C_d}{n}{s-1}$. This is the
  number of assignments having distance {\em exactly} $s$ from a fixed
  assignment $\alpha$. Also, it is the number of vectors $\vec{a} \in
  \{0,\dots,d-1\}^n$ with $\sum_{i=1}^n a_i = s$.  Writing down its
  generating function, we see that $(1 + x + \dots + x^{d-1})^n =
  \sum_{s=0}^{(d-1)n} T(n,s) x^s$. For some $r \in \{0,1,\dots, (d-1)n
  + 1\}$ that maximizes $T(n,r)x^r$, we obtain
  $$
  (1 + x + \dots + x^{d-1})^n  = \sum_{s=0}^{(d-1)n} T(n,s) x^s 
  \leq  ((d-1)n + 1) T(n,r) x^r 
  $$
  Solving fo $T(n,r)$ proves the lemma.  \qed
\end{proof}
We apply this lemma for $x = \frac{1}{k}$ and obtain a certain radius
$r$, for which we construct a code $\mathcal{C}$ of covering
$G$-radius $r$. Combining Lemma~\ref{lower-bound-trinom} 
with (\ref{G-runtime}), we obtain a running time of
$$
\frac{d^n k^r x^r}{(1+x+\dots+x^{d-1})^n} \poly(n) 
= \left(\frac{d(k-1)}{k}\cdot \frac{k^d}{k^d-1}\right)^n  \poly(n) \ ,
$$
and we have proven our main theorem.
\begin{theorem}
  For all $d$ and $k$, there is a deterministic algorithm
  solving $(d,k)$-CSP in time
  $$\left(\frac{d(k-1)}{k}\cdot \frac{k^d}{k^d-1}\right)^n  \poly(n) \ .$$
\end{theorem}
There is one issue we have consistently been sweeping under the rug.
We proved Lemma~\ref{lower-bound-binom} and
Lemma~\ref{lower-bound-trinom}, but never addressed the question what
radius $r$ fulfills the stated bound. For the analysis this does not
matter, since $r$ cancels out nicely.  However, if
we were to implement the algorithm, we would have to choose the right
radius. This is not difficult: In Lemma~\ref{lower-bound-trinom}, the
correct $r$ is the one maximizing $T(n,r)x^r$, and $T(n,r)$ can be
computed quickly using dynamic programming.

\section{Optimality of the Directed Cycle}
\label{section-optimality}

We will show that our analysis cannot be improved by choosing a
different vertex-transitive graph $G$ or a different radius $r$. We
ignore graphs that are not vertex-transitive because we have no idea
on how to upper bounding the running time of $\gcspsb$, not to speak
of estimating the size of a good covering code.\\

Let $G$ be a vertex-transitive graph on $d$ vertices.  For some vertex
$u \in V(G)$, we denote by $d_i$ the number of vertices $v \in V(G)$
having $d_G(u,v)=i$. Since $G$ is finite, the sequence
$d_0,d_1,\dots,$ eventually becomes $0$. Denoting the diameter of $G$
by $s$, it holds that $d_i=0$ for all $i\geq s$. If $G$ is connected
(which we do not necessarily assume), the $d_i$ add up to $d$. Since
$G$ is vertex-transitive, the $d_i$ do not depend on the vertex $u$.
Clearly, $G$ is $d_1$-regular, and $\gcspsb$ runs in time $(d_1 k)^r
\poly(n)$ on a $(d,k)$-CSP formula. How do we estimate
$\vol{G}{n}{r}$? Again we define $T(n,r) = \vol{G}{n}{r} -
\vol{G}{n}{r-1}$, i.e., the number of elements having distance {\em
  exactly} $r$ from some fixed $\alpha$.  The $T_G(n,r)$ obey the
recurrence
$$
T(n,r) = \sum_{i=0}^{s} d_i T(n-1,r-i) \ .
$$
This is easy to see: Fix $\alpha \in [d]^n$. How many $\beta$ are
there such that $d_G(\alpha,\beta)=r$? Consider the first coordinates
$\alpha_1$ and $\beta_1$. If $d_G(\alpha_1, \beta_1)=i$, then there
are $d_i$ ways to choose $\beta_1$, and the distances at the remaining
$n-1$ positions add up to $r-i$. Some moments of thought reveal
the following identity:
$$
\left(\sum_{i=0}^{s}
  d_i x^i\right)^n = \sum_{i=0}^{sn} T(n,i)x^i
$$
Before, we were interested in bounding $\vol{G}{n}{r}$ from below. Now
we want to bound it from above, because we want to argue that any code
$\mathcal{C} \subseteq [d]^n$ of covering radius
$r$ must necessarily be large, and the algorithm must be slow.\\

\begin{lemma}
  For any $n \in \N$, $r \in \{0,1,\dots,sn\}$ and any $x \in
  [0,1]$, it holds that
  $$
  \vol{G}{n}{r} \leq \frac{\left(\sum_{i=0}^{s} d_i x^i\right)^n}{x^r}
    \ .
  $$
\end{lemma}
\begin{proof}
  $\left(\sum_{i=0}^{s} d_i x^i\right)^n =
  \sum_{i=0}^{sn} T(n,i)x^i \geq \sum_{i=0}^{r} T(n,i)x^i
  \geq \sum_{i=0}^{r} T(n,i)x^r 
  = x^r \vol{G}{n}{r}$, and for the last inequality we
  needed that $x \in [0,1]$, thus $x^i \geq x^r$ for $i \leq r$.
  Re-arranging terms yield the claimed bound.
  \qed
\end{proof}
Clearly any code $\mathcal{C} \subseteq [d]^n$ with $\bigcup_{\alpha
  \in \mathcal{C}}\ball{G}{r}{\alpha} = [d]^n$ must satisfy
$$|\mathcal{C}| \geq \frac{d^n}{\vol{G}{n}{r}} \ .$$
 Since $\gcspsb$ takes
time $(kd_1)^r$, the total running time is at least
$$
\frac{d^n}{\vol{G}{n}{r}} (kd_1)^r
\geq \frac{d^n x^r (kd_1)^r} { \left(\sum_{i=0}^{s} d_i x^i\right)^n} \ ,
$$
where this inequality holds for all choices of $x$. Setting 
$x = \frac{1}{kd_1}$, we see that the running time is at least
$$
\frac{d^n}{\left(\sum_{i=0}^{s} d_i k^{-i} d_1^{-i}\right)^n} \ .
$$
In a $d_1$-regular graph, the number of vertices at distance $i$ from $u$
can be at most $d_1^i$. In other words, $d_i \leq d_1^{i}$, and
the above expression is at least
$$
\frac{d^n}{\left(\sum_{i=0}^{s}  k^{-i}\right)^n} \ ,
$$
which, up to a polynomial factor, is the same as what we get 
for the directed cycle on $d$  vertices.

\section{Conclusion and Open Problems}

We can apply the same idea to Sch\"oning's algorithm: When picking a
literal $(x \ne c)$ uniformly at random from an unsatisfied constraint
of $F$ (see line 8 of $\cspsrun$), we choose a new truth value $c'$
uniformly at random from the set $\{c'\in [d] \ | \ (c,c')\in E(G)\}$.
With $G=K_d$, this is the original algorithm $\schoening$, and
surprisingly, for $G$ being the directed cycle, one obtains exactly
the same running time $(d(k-1)/k)^n\poly(n)$. Since the analysis of this
modified $\schoening$ does not introduce any new ideas, we refer the
reader to the appendix and to Andrei Giurgiu's Master's
Thesis~\cite{Giurgiu09}, which presents a general framework for
analyzing random walk algorithms for SAT. Our main open problem is
the following.

\begin{quotation}
  \em For which graph on $d$ vertices does the modified $\cspsrun$ achieve
  its optimal success probability? 
\end{quotation}
  
If we had to, we would guess that no graph can improve Sch\"oning's
algorithm. Intuitively, it does not make sense to restrict the random
choices the algorithm can make, because we have no further information
on which choice might be correct. In the deterministic case, where
every branch is fully searched, it seems to make more sense to
restrict the choices of the algorithm, since this yields an immediate
reduction in the running time of $\gcspsb$.

\ignore{
 Some graphs actually make it
worse: For example, if $d=2^\ell$ is a power to $2$, we can employ $G
= Q_\ell$, the $\ell$-dimensional boolean cube. Some minutes of
thought show that using this graph is equivalent to replacing each
$d$-ary variable by $\ell$ boolean variables, thus transforming a
$(d,k)$-CSP formula with $n$ variables into a $(\ell k)$-CNF formula
with $\ell n$ variables (in fact, I first analyzed this version, then
realized it can be viewed as a graph structure on the set of colors,
and only then did I try to analyze other graphs). Surprisingly,
replacing $d$-ary variables by boolean variables already yields a
modest improvement for
$\dantsin$, but $\schoening$ deteriorates.\\
}

\ignore{
\subsection{Random Boolean Restrictions}

Although most algorithms for $k$-SAT easily generalize to $(d,k)$-CSP,
there is a general way how to obtain a $(d,k)$-CSP algorithm from any
$k$-SAT algorithm $\mathcal{A}$: For a $(d,k)$-CSP formula $F$ on $n$
variables, restrict each variable $x_i$ to a set $\{a,b\} \subseteq
\{1,\dots,d\}$, chosen uniformly at random from all ${d \choose 2}$
pairs and view the restricted formula as a boolean $k$-CNF formula
$F'$. We call $F'$ a {\em random boolean restriction} of $F$. If $F$
is satisfiable, then with probability $(2/d)^n$, $F'$ is as well. We
solve $F'$ using $\mathcal{A}$. If $\mathcal{A}$ has a running time of
$c_k^n \poly(n)$, this yields a Monte-Carlo algorithm for $(d,k)$-CSP
running in time $(dc_k/2)^n \poly(n)$. A quick calculation shows that
for $\schoening$, we obtain the same running time whether we use the
full-fledged $(d,k)$-CSP version or random boolean restriction. This
motivates the following definition: A family $(\mathcal{A}_{d,k})_{d,k
  \in \mathbb{N}}$ of algorithms for $(d,k)$-CSP is {\em essentially
  boolean} if no $\mathcal{A}_{d,k}$ has a better running time than
solving $(d,k)$-CSP by applying $\mathcal{A}_{2,k}$ to a random
boolean restriction.  In general, any family of $(d,k)$-CSP algorithms
with a running time of $(f(k)d)^n \poly(n)$ is essentially boolean, as
for example $\schoening$ and $\dantsin$. Our modified $\dantsin$ using
directed cycles is not essentially boolean, however.\\

Of course, using random boolean restrictions to solve $(d,k)$-CSP
yields a randomized algorithm, so it seems unfair to compare such a
running time with that of $\dantsin$. However, one can efficiently
derandomize the idea of random boolean restrictions, using an idea
similar to covering $[d]^n$ with Hamming balls: A {\em $2$-box} $B
\subseteq [d]^n$ is a set of the form $B = B_1 \times \dots \times
B_n$ where each $B_i \subseteq [d]$ has two elements. Hence a random
boolean restriction can be seen as restricting the solution space to a
randomly chosen $2$-box. In a similar fashion as in~\cite{dantsin}, we
can deterministically construct family $\mathcal{B}$ of $2$-boxes such
that $\bigcup_{B \in \mathcal{B}} B = [d]^n$, and $|\mathcal{B}| \leq
(d/2)^n \poly(n)$~\cite{robin-personal}. For $d$ being even, this is
trivial, and $\mathcal{B}$ can be given explicitly. For odd $d$, we do
not see any simple construction, though. This way, a deterministic
algorithm $\mathcal{A}$ solving $k$-SAT in time $c_k^n \poly(n)$ can
be turned into a {\em deterministic} algorithm solving $(d,k)$-CSP in
time $(dc_k/2)^n \poly(n)$.
}

\section*{Acknowledgments}
Thanks a lot to Emo Welzl and Robin Moser for fruitful and pleasant
discussions.

\bibliographystyle{abbrv}
\bibliography{refs}

\newpage
\appendix

\section{Sch\"oning's Algorithm With Directed Cycles}

Can we apply the same idea to Sch\"oning's algorithm? When picking a
literal $(x \ne c)$ uniformly at random from an unsatisfied constraint
of $F$ (see line 8 of $\cspsrun$), we choose a new truth value $c'$
uniformly at random from $[d]\setminus \{c\}$. We modify this
algorithm as follows: Using a graph $G$ with vertex set $[d]$, we
choose the new color uniformly at random from the set $\{c'\in [d] \ |
\ (c,c')\in E(G)\}$. If $G = K_d$, this is nothing new. What if $G$ is
the directed cycle?  Let the $d$ colors be $0,1,\dots,d-1$ and let the
edges be $(i, i+1)$ (addition taken modulo $d$). This means that we
always change color $c$ to color $c+1$.  Let $\beta$ be a
fixed satisfying assignment and $\alpha$ be the current
(non-satisfying) assignment in the algorithm $\cspsrun$.  If $\beta$
satisfies the literal $(x\ne c)$, i.e. $\beta(x) \ne c$, then changing
the color of $x$ from $c$ to $c+1$ decreases the distance from
$\alpha$ to $\beta$ by $1$.  Otherwise, if $\beta(x)= c$, then the
distance from $\alpha$ to $\beta$ increases by $d-1$. If $C$ is a
constraint involving $k$ literals and which is unsatisfied by
$\alpha$, then with probability at least $\frac{1}{k}$ we choose a
literal that is satisfied by $\beta$, and decrease the distance by
$1$, and with probability at most $\frac{k-1}{k}$, we choose a literal
not satisfied by $\beta$, increasing the distance by $d-1$. To analyze
the algorithm, we define a Markov chain (see Figure~\ref{markov}):
\begin{figure}
  \begin{center}
    \epsfig{file=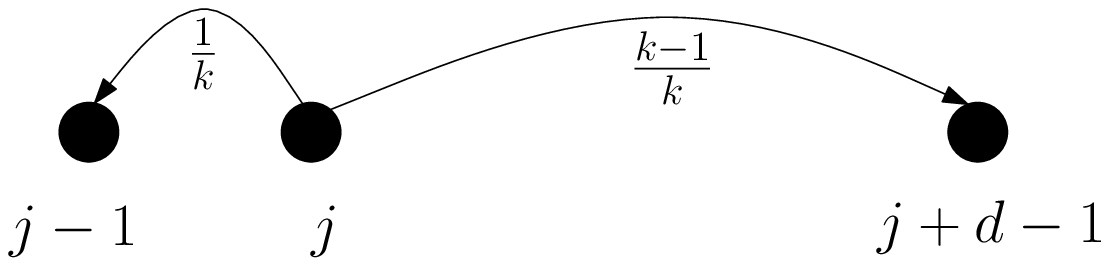,width=0.5\textwidth}
  \end{center}
  \caption{Part of the Markov Chain}
  \label{markov}
\end{figure}
The states of the Markov chain are $\N_0 \cup \{S\}$, with $S$ being a
special starting state. The states $j \in N_0$ represent the distance
from $\alpha$ to some fixed satisfying truth assignment $\beta$.  The
transition probabilities are as follows: For $0 \leq j \leq (d-1)n$,
the probability $p_{S,j}$ of going from $S$ to $j$ is
$\frac{T_G(n,j)}{d^n}$, where $T_G(n,j) = \vol{G}{n}{j} -
\vol{G}{n}{j-1}$ is the number of assignments $\alpha$ such that
$d_G(\alpha,\beta)=j$.  When taking a step from $S$ to some $j$
according to the transition probabilities, $j$ will be distributed
exactly as $d_G(\alpha,\beta)$ for $\alpha \in_{\rm u.a.r.} [d]^n$.
Furthermore, for $j \geq 1$, $p_{j,j-1}$ is $\frac{1}{k}$, and
$p_{j,j+d-1} = \frac{k-1}{k}$, and $p_{0,0}=1$. Here, we only sketch
analysis of this Markov chain. For details, please see Giurgiu's
Master Thesis~\cite{Giurgiu09}. The probability of $\cspsrun$ finding a
satisfying assignment is at least the probability of this Markov chain
reaching state $0$ after at most $cn$ steps, with $c$ being the
constant in line 2 of $\cspsrun$.  As it turns out, the probability that
we reach $0$ in at most $cn$ steps, conditioned on the event that $0$
is reached at all, is rather high.  Note that with positive (in fact,
quite large) probability, we will never reach state $0$. Therefore, to
analyze the success probability of $\cspsrun$, it suffices to lower bound
the probability that our random walk eventually reaches $0$. Let $P_j$
be the probability that a random walk starting in state $j$ eventually
reaches $0$. The $P_j$ obey the equation
\begin{eqnarray}
P_j = \frac{1}{k} P_{j-1} + \frac{k-1}{k} P_{j+d-1} \ .
\label{eqn-Pj}
\end{eqnarray}
Observe that if some $\lambda \in (0,1)$ satisfies
\begin{eqnarray}
\lambda = \frac{1}{k} + \frac{k-1}{k}\lambda^d \ ,
\label{eqn-lambda}
\end{eqnarray}
then $P_j = \lambda^j$ satisfies (\ref{eqn-Pj}). Here we would have to
show that (\ref{eqn-lambda}) has a unique ``reasonable'' solution for
each $d$, and that $\lambda^j$ is in fact the unique solution to
(\ref{eqn-Pj}). We can compute the probability that we eventually
reach $0$:
$$
    \mathbf{\rm P}[0 \textnormal{ eventually reached}]
     = \sum_{j=0}^{(d-1)n} \frac{T_G(n,j)}{d^n} \lambda^j
     = \frac{1}{d^n}(1 + \lambda + \lambda^2 + \dots + \lambda^{d-1})^n \ ,
$$
since $(1+x+x^2+\dots+x^{d-1})^n = \sum_{i=0}^{(d-1)n}T_G(n,i)x^i$.
The above expression involves a geometric series and thus equals
$\left(\frac{\lambda^d - 1}{d (\lambda - 1)}\right)^n$.
From (\ref{eqn-lambda}) we learn that $\lambda^d = \frac{k\lambda - 1}{k-1}$,
and plugging this into the previous expression yields
$$
\left(\frac{\lambda^d - 1}{d (\lambda - 1)}\right)^n = 
\left(\frac{\frac{k\lambda - 1}{k-1} - 1}{d(\lambda-1)}\right)^n
= \left(\frac{k}{d(k-1)}\right)^n \ .
$$
Now the probability that $\cspsrun$ finds a satisfying assignment is at
least $\left(\frac{k}{d(k-1)}\right)^n \frac{1}{\poly(n)}$, and if we
repeat it $\left(\frac{d(k-1)}{k}\right)^n \poly(n)$ times, we find a
satisfying assignment with constant probability (if one exists). This
is exactly the running time of Sch\"oning's algorithm we got before.
Hence we see: Running Sch\"oning with $G$ being $K_d$ or being the
directed cycle makes no difference.
\ignore{
If $d = 2^\ell$ is a power of $2$,
running it with $G = Q_\ell$ being the $\ell$-dimensional Boolean cube
actually makes the running time worse (as taking $G = Q_\ell$ is the
same as replacing each $d$-ary variable by $\ell$ boolean
variables).

This yields the following obvious question, which 
we currently cannot answer:

\begin{quotation}
  \em For which graph on $d$ vertices does the modified $\cspsrun$ achieve
  its optimal success probability? Is there a graph that works better than
  $K_d$?
\end{quotation}
  
We think that no graph can improve Sch\"oning's algorithm.
Intuitively, it does not make sense to restrict the random choices the
algorithm can make, because we have no further information on which
choice might be correct. In the deterministic case, where every branch
is fully searched, it seems to make more sense to restrict the choices
of the algorithm, since this yields an immediate reduction in the
running time of $\gcspsb$.}

\section{Constructing the Covering Code}

We show how to deterministically construct a code
$\mathcal{C}\subseteq [d]^n$ of covering radius $r$, i.e.,
$\bigcup_{\alpha \in \mathcal{C}} \ball{G}{r}{\alpha} = [d]^n$, for
$G$ being the directed cycle on $d$ vertices. The construction is just
a generalization of the one in Dantsin et al.~\cite{dantsin}.
\begin{lemma}
  Let $G$ be the directed cycle on $d$ vertices. For any $n \in \N$
  and $x \geq$, there is an $r \in \{0,\dots, (d-1)n\}$ such that
  $$\vol{G}{n}{r} \geq \frac{1}{(d-1)n+1} 
  \frac{(1+x+x^2+\dots+x^{d-1})^n}{x^r} \ ,$$  
  and there is a code $\mathcal{C} \subseteq [d]^n$ of size
  at most 
  $$
  \frac{[d]^n x^r}{(1+x+x^2+\dots+x^{d-1})^n} \poly(n) 
  $$
  which can be constructed deterministically in time
  $O(\mathcal{|C|})$.
\end{lemma}
\begin{proof}
  The proof idea is as follows: A probabilistic argument shows that a
  code $\mathcal{C}^*$ of claimed size exists (one obtains
  $\mathcal{C}^*$ by sampling random points in $[d]^n$), and then one
  invokes a greedy polynomial time approximation algorithm for the Set
  Cover problem (see~\cite{Hochbaum}, for example).  This returns a
  code of size at most $|\mathcal{C}^*| \poly(n)$. The problem is that
  this instance of Set Cover has a ground set of size $d^n$, and $d^n$
  sets to choose from, thus the approximation algorithm will take at
  least $d^n$ steps. As in Dantsin et al.~\cite{dantsin}, we solve
  this problem by partitioning our $n$ variables into $b$ blocks of
  length $n/b$ each, where $b$ is a constant, depending on $d$ but not $n$.\\

  Let us be more formal. We first construct a covering code for
  $[d]^{n/b}$. By Lemma~\ref{lower-bound-trinom}, we know that for any
  $x \geq 0$, there is an $r \in \{0,\dots,(d-1)n/b\}$ such that
  $$
  \vol{G}{n/b}{r} \geq \frac{1}{(d-1)n+1}\frac{(1+x+\dots+x^{d-1})^{n/b}}
  {x^r} \ .
  $$
  Using this $r$, we choose a set $\mathcal{C}^* \subseteq [d]^{n/b}$
  by randomly sampling $\frac{\ln(d^{n/b}) d^{n/b}}{\vol{G}{n/b}{r}}$
  elements from $[d]^{n/b}$, uniformly at random with replacement.
  This is only a feature of the proof -- the sampling is not part of
  our deterministic construction. For any fixed $\beta \in
  [d]^{n/b}$, it holds that
  $$
  {\rm P}[\beta \not \in \bigcup_{\alpha \in \mathcal{C}^*} 
  \ball{G}{r}{\alpha}] = 
  \left(1 - \frac{\vol{G}{n/b}{r}}{d^{n/b}}\right)^{|\mathcal{C^*}|} 
  < e^{-\ln(d^{n/b})} = d^{-n/b} \ .
  $$
  By the union bound, we see that with non-zero probability, no
  assignment $\beta$ is uncovered, and thus there exists a code
  $\mathcal{C}^*$ of desired size and covering radius $r$. We
  construct an instance of Set Cover: The ground set is $[d]^{n/b}$,
  and the set system consists of all $\ball{G}{r}{\alpha}$ for $\alpha
  \in [d]^{n/b}$. The deterministic polynomial-time approximation
  algorithm will in time $\poly(d^{n/b})$ find a code $\mathcal{C}
  \subseteq [d]^{n/b}$ of size $O(|\mathcal{C}^*| n)$. We define
  $\mathcal{C}' \subseteq [d]^n$ by $\mathcal{C}' := \mathcal{C}^b$,
  the $b$-fold Cartesian product. It is easy to see that
  $$
  \bigcup_{\alpha \in \mathcal{C}'} \ball{G}{rb}{\alpha} = [d]^n 
  $$
  and 
  $$
  |\mathcal{C}'| = |\mathcal{C}|^b \leq 
  \frac{d^n}{x^{rb}}{(1+x+\dots+x^{d-1})^n} \poly(n)^b \ .
  $$
  By choosing $b$ large enough, although still constant, we can make
  sure that the running time of the approximation algorithm is at most
  $|\mathcal{C}'|$. This concludes the proof.
  \qed
\end{proof}
Actually the proof works as well for arbitrary vertex-transitive
graphs, not only directed cycles, but the formulas become uglier.

\end{document}